\documentclass[12pt,published]{article}
\usepackage{jheppub}
\usepackage{amsthm, amsmath, amsfonts, amsxtra, amssymb, appendix, euscript, mathrsfs, MnSymbol, verbatim, enumerate, multicol, multirow, color, babel, tikz, tikz-cd, tikz-3dplot, array, enumitem, hyperref, thm-restate, thmtools, datetime, graphicx, tensor, braket, slashed, adjustbox, mathtools, etoolbox, bbding, esint, pgfplots, ytableau, float, rotating, mathdots, savesym, wasysym, amscd, pifont, setspace, cleveref, wrapfig, picture, subfigure, tabularx, youngtab, parskip}
\usepackage[normalem]{ulem} 
\usepackage[utf8]{inputenc}
\usepackage[all]{xy}
\numberwithin{equation}{section}
\numberwithin{figure}{section}
\usetikzlibrary{arrows, positioning, decorations.pathmorphing, decorations.markings, decorations.pathreplacing, decorations.markings, matrix, patterns, snakes}
\tikzset{
  on each segment/.style={
    decorate,
    decoration={
      show path construction,
      moveto code={},
      lineto code={
        \path [#1]
        (\tikzinputsegmentfirst) -- (\tikzinputsegmentlast);
      },
      curveto code={
        \path [#1] (\tikzinputsegmentfirst)
        .. controls
        (\tikzinputsegmentsupporta) and (\tikzinputsegmentsupportb)
        ..
        (\tikzinputsegmentlast);
      },
      closepath code={
        \path [#1]
        (\tikzinputsegmentfirst) -- (\tikzinputsegmentlast);
      },
    },
  },
  mid arrow/.style={postaction={decorate,decoration={
        markings,
        mark=at position .7 with {\arrow[#1]{stealth}}
      }}},
}

\hypersetup{colorlinks=true}
\hypersetup{linkcolor=black}
\hypersetup{citecolor=black}
\hypersetup{urlcolor=black}

\makeatletter
\def\oversortoftilde#1{\mathop{\vbox{\m@th\ialign{##\crcr\noalign{\kern3\p@}%
				\sortoftildefill\crcr\noalign{\kern3\p@\nointerlineskip}%
				$\hfil\displaystyle{#1}\hfil$\crcr}}}\limits}

\def\sortoftildefill{$\m@th \setbox\z@\hbox{$-$}%
	\braceld\leaders\vrule \@height\ht\z@ \@depth\z@\hfill\braceru$}

\makeatother

\theoremstyle{plain}
\newtheorem*{thm*}{Theorem}
\newtheorem{thm}{Theorem}[section]

\newtheorem{lem}[thm]{Lemma}
\newtheorem{prop}[thm]{Proposition}

\theoremstyle{definition}
\newtheorem{defn}[thm]{Definition}
\newtheorem*{defn*}{Definition}

\newtheorem{rem}[thm]{Remark}

\newtheorem*{thm:maintheoremCstar}{Theorem \ref{thm:maintheoremCstar}}

\crefname{lemma}{lemma}{lemmas}
\Crefname{lemma}{Lemma}{Lemmas}
\crefname{thm}{theorem}{theorems}
\Crefname{thm}{Theorem}{Theorems}
\crefname{defn}{definition}{definitions}
\Crefname{defn}{Definition}{Definitions}

\DeclarePairedDelimiterX{\abs}[1]{\lvert}{\rvert}{\ifblank{#1}{{}\cdot{}}{#1}}

\makeatother

\newcommand{\calh}{\mathcal{H}}

\title{\boldmath Holographic tensor networks from hyperbolic buildings}
\author{Elliott Gesteau,}
\author{Matilde Marcolli,}
\author{Sarthak Parikh}
\affiliation{Division of Physics, Mathematics, and Astronomy, California Institute of Technology,\\Pasadena, CA 91125, U.S.A.}
\emailAdd{egesteau@caltech.edu}
\emailAdd{matilde@caltech.edu}
\emailAdd{sparikh@caltech.edu}
\abstract{We introduce a unifying framework for the construction of holographic tensor networks, based on the theory of hyperbolic buildings. The underlying dualities relate a bulk space to a boundary which can be homeomorphic to a sphere, but also to more general spaces like a Menger sponge type fractal. In this general setting, we give a precise construction of a large family of bulk regions that satisfy complementary recovery. For these regions, our networks obey a Ryu--Takayanagi formula. The areas of Ryu--Takayanagi surfaces are controlled by the Hausdorff dimension of the boundary, and consistently generalize the behavior of holographic entanglement entropy in integer dimensions to the non-integer case. Our construction recovers HaPPY--like codes in all dimensions, and generalizes the geometry of Bruhat--Tits trees. It also provides examples of infinite-dimensional nets of holographic conditional expectations, and opens a path towards the study of conformal field theory and holography on fractal spaces.}

\pgfplotsset{compat=1.17}
\begin{document}
\maketitle
\flushbottom

\section{Introduction}

Holographic tensor networks are one of the main tools to model the emergence of spacetime in the AdS/CFT correspondence, and the associated error-correcting structure. Since the discovery of the HaPPY code \cite{Pastawski:2015qua}, a plethora of examples have shown that holographic tensor networks make it possible to model numerous aspects of holography, in particular, the Ryu--Takayanagi formula \cite{Ryu:2006bv, Faulkner:2013ana}  and quantum error correction \cite{Harlow:2016vwg}.

Most holographic tensor networks realize an exact or approximate quantum error-correcting code, in the sense that they (almost) isometrically map a Hilbert space of semiclassical bulk degrees of freedom, modelled by qudits  associated to dangling legs, to a boundary Hilbert space, which represents the Hilbert space of the boundary conformal field theory, and is modelled by qudits on the boundary of the tensor network. For well-chosen boundary regions, the tensor network satisfies the quantum Ryu--Takayanagi formula \cite{Pastawski:2015qua}. This means that the entanglement entropy of the restriction to a given boundary region of a state in the code subspace equals the sum of the entanglement entropy of that state in the bulk and an area contribution proportional to the number of internal legs of the tensor network cut by the Ryu--Takayanagi surface of the region. The latter is defined as the surface delimiting the region attainable from the boundary by the greedy algorithm \cite{Pastawski:2015qua}. Complementary recovery is then achieved when the greedy algorithm reaches complementary bulk regions from complementary parts of the boundary.

Holographic tensor networks have an interesting geometric structure, which ranges from tesselations of the hyperbolic plane \cite{Pastawski:2015qua} or higher-dimensional spaces \cite{Kohler:2018kqk} to $p$-adic spaces \cite{Heydeman:2018qty} (see \cite{Chen:2021ipv} for a complimentary perspective). Although it is often briefly mentioned that hyperbolic tesselations have to do with Coxeter systems \cite{Kohler:2018kqk}, the interplay between holographic tensor networks and hyperbolic geometry has been left largely unexplored.\footnote{However, see~\cite{Bhattacharyya:2016hbx} for tensor network constructions in terms of tessellations of the hyperbolic plane inspired by Coxeter systems and \cite{Kohler:2018kqk} for a geometric condition for a tiling to define an isometry.} The goal of this paper is to investigate this link in more detail, and show the pertinence of the language of Gromov-hyperbolicity and building theory to talk about holographic tensor networks.

More precisely, one can ask the following questions:
\begin{itemize}
    \item What are the geometric structures that underlie the construction of holographic tensor networks?
    \item How can we know when a graph makes a good holographic quantum error-correcting code?
    \item Can we construct tensor networks that model holographic dualities where the boundary is not homeomorphic to a sphere?
    \item How can we predict which boundary regions satisfy complementary recovery?
    \item How do we take an infinite-dimensional limit of holographic tensor networks?
    \item What are some insights that the geometry of holographic tensor networks can provide for studying full AdS/CFT?
\end{itemize}

In this paper, we will provide an answer to these questions in the case of networks constructed out of perfect tensors, utilizing the framework of Gromov hyperbolicity and hyperbolic buildings. It will turn out that known examples of holographic tensor networks can all be described by the notion of hyperbolic building. Buildings, which can have different geometric structures (hyperbolic, Euclidean, spherical) are a geometric construct originally introduced by Jacques Tits \cite{BuildingBook} aimed at geometrizing some aspects of group theory. In order to do geometry on hyperbolic buildings, the right toolkit is provided by Gromov's theory of hyperbolicity, which studies metric spaces whose distance has a particular property that can be viewed as a more abstract formulation of the concept of negative curvature. On such spaces, a notion of boundary at infinity, the Gromov boundary, can be defined. In our context, holographic dualities will be between semiclassical theories living on a hyperbolic building and theories living on its Gromov boundary. Tensor networks will contain bulk dangling legs in the chambers of the buildings, and boundary legs at a certain cutoff of the building. These boundary legs can be thought of as a coarse-grained approximation to the Gromov boundary.

One striking feature of this approach is that, unlike the case of the HaPPY code or similar 
tessellations, the Gromov boundary of most hyperbolic buildings is not isomorphic to a sphere -- rather, it is often isomorphic to a fractal, which has a much more intricate geometric structure. Thus, our holographic dualities provide examples of dualities where the boundary theory lives on a sphere (or more generally a homology sphere) of any dimension, but also where it lives on more complicated spaces of non-integer dimension. Moreover, as we shall see, there is a close analog of conformal invariance for these theories, and it exhibits a behavior closely related to that of a CFT.

Using the more general point of view of hyperbolic buildings also allows us to reflect on some features of known holographic tensor networks, like complementary recovery. While this property is supposed to hold for arbitrary boundary regions up to nonperturbative errors in $G_N$ in AdS/CFT \cite{Dong:2017xht}, it does not hold for all regions of the HaPPY code \cite{Pastawski:2015qua}. While this fact may seem puzzling, it has a very natural explanation in terms of hyperbolic buildings: their global symmetry groups are smaller than the conformal group, hence only regions which are well-adapted to these symmetry groups will satisfy complementary recovery. In particular, we will provide a systematic construction of regions in the network that do satisfy complementary recovery, and applies to all networks with perfect tensors. 

For these regions, we will show that a Ryu--Takayanagi formula holds, and that for ball-shaped boundary regions the number of links on the Ryu--Takayanagi surface follows a logarithmic law in the radius when the boundary time slice has dimension 1, and a power law in the radius, with exponent the \textit{Hausdorff dimension of the boundary minus one} when the boundary time slice has dimension greater than 1. This recovers known results for the scaling of entanglement entropy in traditional conformal field theory, and introduces some new scalings that are very suggestive: the Ryu--Takayanagi surfaces see the fractal structure of the boundary! It is then of course very tempting to speculate that networks with fractal boundary simulate conformal field theories on fractal spaces.

Another interesting aspect of our approach is that the Gromov boundary lives at infinity, and makes it very natural to define an infinite-dimensional limit of holographic tensor networks. This problem has already been touched upon in \cite{Gesteau:2020hoz, Gesteau:2020rtg}, and more generally, the question of describing holographic quantum error-correction in the language of infinite-dimensional operator algebras is an active research program \cite{Kang:2018xqy, Gesteau:2020rtg, Faulkner:2020hzi, Gesteau:2021tba}. Here we will see that there is an appropriate way to take the limit of our holographic codes such that they give a net of holographic conditional expectations. This structure has been recently shown to capture essential aspects of bulk reconstruction in AdS/CFT \cite{Faulkner:2020hzi}. 

We now give the main results of this work:
\begin{itemize}
    \item A general framework for the construction of holographic tensor networks with perfect tensors is given in terms of building theory and Gromov hyperbolic spaces.
    \item This allows us to recover all known constructions, and gives examples of holographic tensor networks in all integer dimensions.
    \item A lot of tensor networks fitting into that framework also have a fractal boundary, hinting at new holographic dualities where the boundary theory lives on a fractal.
    \item A condition for our networks to be isometric is given, as well as a construction of regions that satisfy complementary recovery. This construction can be applied to all the buildings examined in this paper.
    \item A Ryu--Takayanagi formula is proven, showing that boundary entanglement entropy for ball-shaped regions follows a logarithmic law in the radius when the boundary has dimension 1, and a power law in the radius when the boundary has higher Hausdorff dimension, where the exponent is the Hausdorff dimension of the boundary minus one.
    \item A general technique is given to construct an infinite-dimensional limit of our networks in the language of operator algebras. This limit gives rise to a net of holographic conditional expectations.
    \item All our results can be applied to known examples of holographic tensor networks with perfect tensors, such as the HaPPY code.
\end{itemize}

The rest of the paper is organized as follows: In Section~\ref{GrSec}, we recall the basics of the theory of Gromov-hyperbolic spaces and their boundaries, as well as the notions of hyperbolic groups and buildings. In Section~\ref{BourdonSec}, we focus for clarity on a particular case: that of Bourdon buildings, which can be understood as HaPPY codes with branching. These buildings have a boundary homeomorphic to a fractal Menger sponge, and we show that the resulting tensor networks satisfy complementary recovery for nice regions, a Ryu--Takayanagi formula with the expected scaling in terms of the Hausdorff dimension of the boundary. We also construct an infinite-dimensional limit for these tensor networks. In Section~\ref{HighDimSec}, we extrapolate the methods of the previous section to define holographic quantum error correcting codes on a much larger class of higher-dimensional hyperbolic buildings. We give some explicit examples, including ones where the boundary is a homology-sphere of arbitrary integer dimension. In Section~\ref{DiscussSec}, we briefly comment on the results and discuss some potential future directions.
In Appendix~\ref{APPI53} we present an explicit example of a Bourdon building, and in the slightly more technical Appendix~\ref{PStheory} we summarize relevant results from Patterson-Sullivan theory which help formalize aspects of conformal field theories on fractal spaces.

\section{Gromov-hyperbolic spaces, hyperbolic groups and buildings}\label{GrSec}
In this section, we introduce the general notion of Gromov-hyperbolic space, which will be underlying our choices of bulk spaces. We focus on two types of Gromov-hyperbolic spaces, hyperbolic groups and hyperbolic buildings, which will be the ones we will use in order to construct our examples of holographic duality.
\subsection{General definitions}
Gromov-hyperbolic spaces naturally generalize the setup in which one usually considers holographic dualities. We begin with a definition of the Gromov product, which is the crucial ingredient in the definition of Gromov-hyperbolic spaces (see \cite{CoDePa}):
\begin{defn}
For $(X,d)$ a metric space, the Gromov product of two points $y,z\in X$ with respect to $x\in X$ is given by $$(y,z)_x:=\frac{1}{2}(d(x,y)+d(x,z)-d(y,z)).$$
\end{defn}
From there, a Gromov-hyperbolic space is defined by the following condition:
\begin{defn}
Let $\delta>0$. $(X,d)$ is said to be $\delta$-hyperbolic if for all $x,y,z,w\in X,$ $$(x,z)_w\geq \mathrm{min}((x,y)_w,(y,z)_w)-\delta.$$ $X$ is then Gromov-hyperbolic if it is $\delta$-hyperbolic for some $\delta>0$.
\end{defn}
For our purposes, one of the main interesting features of Gromov-hyperbolic spaces is that they are endowed with a natural notion of boundary, which will make it possible for us to formulate a bulk-to-boundary correspondence. We first need to formulate a notion of geodesics in Gromov-hyperbolic spaces.

\begin{defn}
Fix an origin $O\in X$. A geodesic ray in $X$ is an isometry $r: [0,+\infty)\longrightarrow X$ such that $$r(0)=O$$ and for all $t>0$, $r([0,t])$ is the shortest path from $O$ to $r(t)$ in X. Two geodesic rays $r_1$ and $r_2$ are said to be equivalent if there exists $K>0$ such that for all $t>0$, $$d(r_1(t),r_2(t))\leq K.$$The Gromov boundary $\partial X$ of $X$ is defined to be the set of equivalence classes of geodesic rays starting at $O$.
\end{defn}

$X\cup\partial X$ is then endowed with a natural topology: a basis is given by the sets of the form $$V(p,\rho):=\{q\in \partial X, \mathrm{there\;exist\;geodesic\;rays\;}r_1,\;r_2\;\mathrm{ending\;at\;}p,q\;\mathrm{such\;that\;}\underset{t_1,t_2\rightarrow +\infty}{\mathrm{lim}}(r_1(t_1),r_2(t_2))_O\geq \rho.\} $$

A convenient feature of the Gromov boundary $\partial X$ is that one can construct a natural metric on it \cite{Coornaert}. In what follows, we shall fix a base point $O$ for $X$. For $x\in X$, define $$|x|:=d(x,O).$$ Then for $a>1$, define $$|x-y|_a:=\underset{r\; \mathrm{path\; from\;}x\;\mathrm{to}\;y}{\mathrm{inf}}\int_r a^{-|x|} dx.$$
One can then prove \cite{Coornaert} that there exists $a_0>1$ such that for $1<a<a_0$, $X\cup\partial X$ is homeomorphic to the completion of $X$ for the metric $|\cdot|_a$. Moreover, the metric $|\cdot|_a$ is very well-controlled by the Gromov product of $X$ (we shall see explicit examples of this in what follows). In particular, there exists a constant $\lambda$ which only depends on $\delta$ and $a$ such that if $\xi$ and $\eta$ are on $\partial X$, for $x$ and $y$ in small enough neighborhoods of $\xi$ and $\eta$, we have \cite{Coornaert}: 
\begin{equation}\label{distGr}
\lambda^{-1} a^{-(x,y)_O}\leq |\xi - \eta|_a \leq \lambda a^{-(x,y)_O}.
\end{equation}
Hence $|\cdot|_a$ is a natural metric, called \textit{visual metric}, on $\partial X$. From here on, we will consider the choices of the base point $O$ and of a given $a>1$ as implicit and
we will drop the explicit notation unless needed.\footnote{For instance, from here on we will drop the subscript  $a$ in the metric $|\cdot|_a$.}

Most of the Gromov-hyperbolic spaces we will consider here can be interpreted in terms of hyperbolic groups. Here we give a general definition of a hyperbolic group, of which we will consider specific explicit examples in the subsequent sections.

\begin{defn}
Let $G$ be a finitely generated group, and $X$ be its Cayley graph. $G$ is said to be a hyperbolic group if $X$, endowed with its graph metric, is Gromov-hyperbolic.
\end{defn}

\subsection{Hyperbolic buildings}

In this paper, the main setup will be that of hyperbolic buildings. We will also encounter their isometry groups, which will turn out to be hyperbolic groups. The theory of buildings is a rich and fruitful mathematical framework, first introduced by Tits \cite{BuildingBook}, whose goal is to geometrize notions of group theory. Here we introduce the basic terminology associated to this theory, which will be utilized in our paper. We will mostly follow the presentation of \cite{Thomas:2016}, which the interested reader can consult for a more thorough introduction.

The simplest examples of buildings are Coxeter systems:

\begin{defn}
A \textit{Coxeter group} is a group $W$ that admits a presentation of the form $$W=\langle s\in S\, |\,  (st)^{m_{st}}=1 \mbox{ for } s,t\in S \rangle \, ,$$ with $S$ a finite set, $m_{ss}=1$ for $s\in S$, $m_{st}$ a (possibly infinite) integer $\geq 2$ for $s\neq t$. The pair $(W,S)$ is then called a \textit{Coxeter system}.
\end{defn}

A particularly nice class of Coxeter systems is given by the \textit{right-angled} ones, which correspond to the case where the $m_{st}$ for $s\neq t$ are either $2$ or $\infty$.

  \begin{figure}
 \begin{center}
  \includegraphics[scale=0.5]{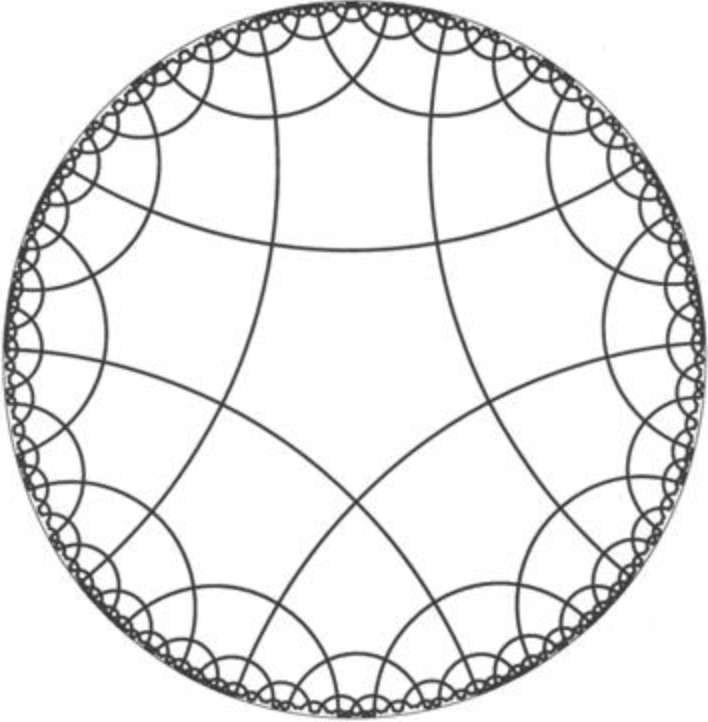}
 \end{center}
 \caption{The right-angled pentagon tiling of the hyperbolic plane ${\mathbb H}^2$ used in the construction of the HaPPY code. \label{FigPentagon}}
 \end{figure}

Many interesting Coxeter systems arise from regular tessellations of $n$-spheres, $n$-Euclidean space, or $n$-hyperbolic spaces. In this case, the Coxeter group is generated by reflections with respect to the sides of the basic  polyhedron. 
This is for example the case for the HaPPY code, which corresponds to a right-angled Coxeter system generated by a right-angled regular pentagon in the hyperbolic plane (see Figure~\ref{FigPentagon}). 
More generally, a convex polyhedron in $\mathbb{X}^n$ (which can be the $n$-sphere, the $n$-Euclidean space or the $n$-hyperbolic space), with all dihedral angles submultiples of $\pi$ is called a \textit{Coxeter polytope}. These three possible cases for $\mathbb{X}^n$ are,
respectively, referred to as {\em spherical}, {\em Euclidean}, and {\em hyperbolic} buildings. 

Another important notion in order to define a building is that of \textit{polyhedral complex}. Without going into too much generality, they are constructed by
gluing polyhedra in spheres, Euclidean space or hyperbolic space, using isometries along the faces. The main object of interest in a polyhedral complex is its \textit{link} at each vertex $x$: it is the $(n-1)$-dimensional polyhedral complex obtained by intersecting the given polyhedral complex with an $n$-sphere of sufficiently small radius centered at $x$. For example, for a 2-complex, the link at vertex $x$ is a graph whose edges correspond to the faces adjacent to $x$, and whose vertices correspond to the edges incident on $x$.

We are now ready to define a hyperbolic building, utilizing the notions of Coxeter polytopes and polyhedral complexes.

\begin{defn}
Let $P$ be a hyperbolic Coxeter polytope, and let $(W,S)$ be the associated Coxeter system. A \textit{hyperbolic building} of type $(W,S)$ is a polyhedral complex $\Delta$ with a maximal family of subcomplexes, called \textit{apartments}, such that they each are isometric to a tessellation of $\mathbb{H}^n$ by copies of $P$ called \textit{chambers}, and \begin{itemize}
    \item Any two chambers of $\Delta$ are contained in a common apartment.
    \item Between any two apartments, there exists an isometry that fixes their intersection.
\end{itemize}
\end{defn}

In the case of a hyperbolic building, $W$, called the Weyl group of the building, is a hyperbolic discrete subgroup of the isometry group of $\mathbb{H}^n$. It is also possible to show that the link of an $n$-dimensional hyperbolic building at each vertex is an $(n-1)$-dimensional spherical building.  
The link structure of our buildings will be crucial for our proof of complementary recovery.

\begin{rem}
The condition that the apartments are tessellations of $\mathbb{H}^n$ can be relaxed to include a larger class of hyperbolic buildings, where the apartments are certain polyhedral complexes (Davis--Moussong complexes) with a hyperbolic structures, which are not tessellations of a single hyperbolic space. This generalization makes it possible to obtain hyperbolic buildings in arbitrary dimension, and we will discuss it in Section~\ref{HighDimSec}.
\end{rem}

\section{Holography on Bourdon buildings}\label{BourdonSec}

Our first example of building holography is given by the case where the bulk space is the Bourdon building $I_{p,q}$ \cite{BourdonPaper}, whose boundary is a Menger sponge, which is universal among topological spaces of topological dimension 1, in the sense that all of them are homeomorphic to a subset of the Menger sponge. When interpreted as a holographic tensor network, we will see that our bulk space gives rise to the expected properties of holographic codes and states: complementary recovery, and a Ryu--Takayanagi formula involving the Patterson--Sullivan measure on the boundary.

\subsection{Bourdon buildings}

We start by defining the building $I_{p,q}$, for $p\geq 5$ and $q\geq 3$, closely following \cite{BourdonPaper}. For the case $q=2$ see Remark~\ref{q2rem} below.

\begin{defn}\label{BourdonDef}
Let $p,q$ be two integers with $p\geq 5$ and $q\geq 3$.
The Bourdon building $I_{p,q}$ is the only simply connected cellular 2-complex such that its 2-cells are isometric to a regular $p$-gon, are attached by their edges and vertices, two 2-cells share at most one edge or one vertex, and the link of each vertex is the bipartite graph $K(q,q)$.
\end{defn} 

The existence and uniqueness of such a building is proven in \cite{BourdonPaper}. Let $\Gamma_{p,q}$ be the hyperbolic group defined by its presentation
$$\Gamma_{p,q}:=\langle s_1,...,s_p\,| \, s_i^q=1,[s_i,s_{i+1}]=1 \rangle.$$
Then $\Gamma_{p,q}$ acts simply transitively on the set of chambers of $I_{p,q}$. In particular, if one fixes a zero chamber in $I_{p,q}$, then the chambers of $I_{p,q}$ can be seen as the words formed by $\Gamma_{p,q}$ up to the relations in the group presentation, where each letter applies a group transformation to the chamber.

  \begin{figure}
 \begin{center}
  \includegraphics[scale=0.5]{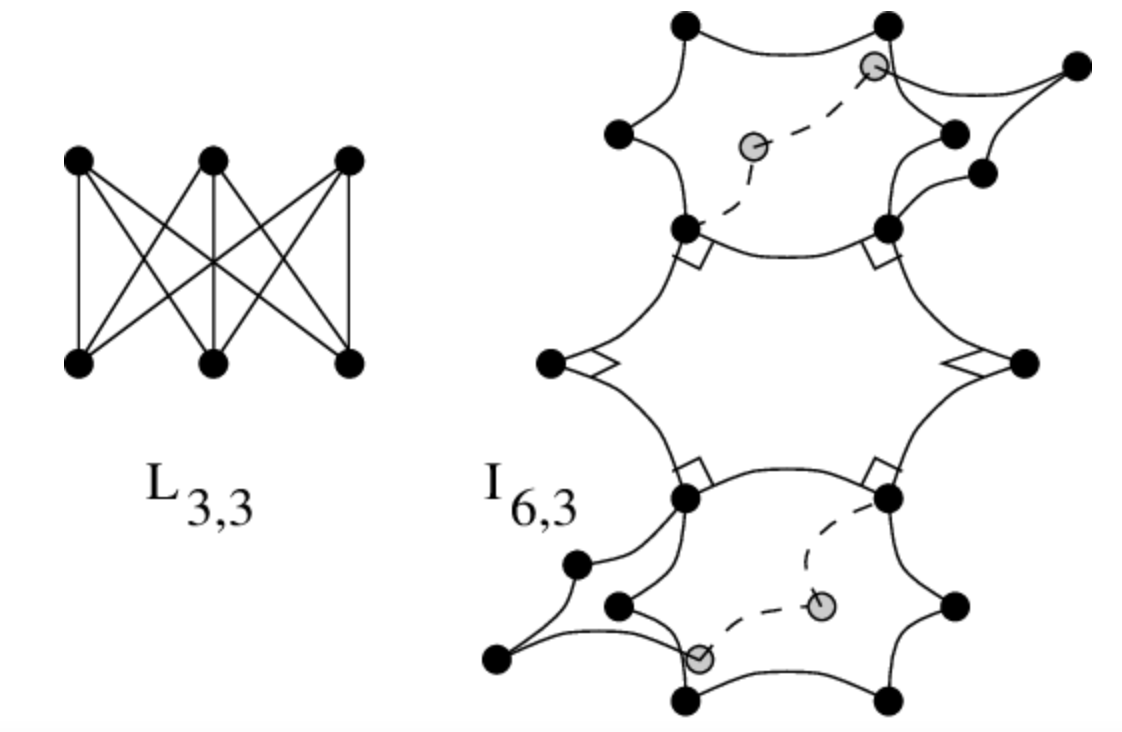}
 \end{center}
 \caption{An example of a subset of a Bourdon building with the associated link, from \cite{BouPaj}. \label{FigBourdon}}
 \end{figure}

\begin{rem}\label{q2rem}
In the case where $q=2$, $\Gamma_{p,q}$ is just a tesselation of the hyperbolic plane by $p$-gons (and the Gromov boundary is a circle). In particular, whenever this tessellation gives rise to an isometric network of perfect tensors, it can be interpreted as a form of HaPPY tiling. We refer to these tessellations as the $p$-gon HaPPY tilings, with $p\geq 5$.
\end{rem}

For $q\geq 3$, the Gromov boundary changes, and becomes a Menger sponge. The Bourdon building then has branching: each tile branches out into $q-1$ tiles at each edge. Nevertheless, the building still possesses a HaPPY-like structure, as its apartments are now tilings of the hyperbolic plane by $p$-gons. The $q=2$ case corresponds to when the building only has one apartment. This nice apartment structure will greatly simplify our analysis of tensor networks on Bourdon buildings, as it will enable us to transpose a lot of useful error-correcting properties of the HaPPY code to the Bourdon case.

\subsection{Bourdon tensor networks}

We now define our Bourdon tensor network. The idea is to extend the HaPPY construction to the more general case of the building $I_{p,q}$. In the rest of this section, we will consider the case of holographic codes with a nontrivial code subspace, and we will hence suppose for simplicity that $p$ is odd, in order for our bulk tensors to always have an even number of legs, independently of $q$.

\begin{defn}\label{BtensorDef}
The $I_{p,q}$ tensor network is constructed in the following way: 
\begin{itemize}
\item Insert a perfect tensor in the chambers of the building.
\item Perform an index contraction between every two chambers sharing an edge. 
\item Add $q-1$ dangling legs for each tensor. 
\item Choose a central tile, and cut the building at a finite distance $\Lambda$ from this central tile. 
\end{itemize}
As usual for holographic quantum error correcting codes, the bulk Hilbert space is identified with the tensor product of the bulk dangling leg 
Hilbert spaces, while the boundary Hilbert space is identified with the tensor product of the boundary leg Hilbert spaces.
\end{defn}

\begin{rem}
The existence of perfect tensors that work for any given choice of $p$ and $q$ in the range of Definition~\ref{BourdonDef} follows from \cite{HuWy}, see
also \cite{HuElSiGu}.
\end{rem}

\begin{rem}
We could have allowed for different types of tensor networks. In particular, when $p$ increases, we could have introduced $k(q-1)$ dangling legs in the bulk with $k\leq p-4$ (as long as the total number of legs is even), while still satisfying the isometry and entanglement wedge reconstruction properties proven in the next sections. 
\end{rem}

\subsection{Bulk-to-boundary isometry}

We first prove that our network defines an isometry at each layer.

\begin{thm}
At each layer, the Bourdon tensor network of Definition~\ref{BtensorDef} determines an isometry from the code subspace to the physical Hilbert space.
\end{thm}

\begin{proof}
In this proof, we will freely use the fact that the $p$-gon HaPPY tilings define an isometry from the bulk to the boundary. For a proof, see \cite{Pastawski:2015qua}. This being said, our strategy will be to introduce an acyclic orientation on the edges of the network such that each tensor has at least as many outgoing links as incoming links (including bulk nodes). We will do this by induction on the layer number $n$. At the first layer, only the bulk node is introduced. We assign an outgoing orientation to all the edges, and the condition is clearly satisfied. Now suppose that up to layer $n$, the network determines an isometry. Take a chamber of the building at layer $n+1$. This chamber can be included in an apartment, which looks like a HaPPY tiling. Inside this apartment, there are more tiles touching our chamber that are part of layer $n+2$ than there are that are part of layer $n$ or $n+1$. By symmetry of the building, for a given tile at layer $n+2$, there are $q-2$ other tiles which are also at layer $n+2$ and share the same edge with our chamber. We therefore define an orientation on the network in the following way: if two chambers share an edge, and one is in a higher layer than the other, then define the orientation of the network from the one in the lower layer to the one in the higher layer. Then collect all adjacent tiles that are in the same layer, and define any acyclic orientation on the corresponding subgraph. This gives a well-defined orientation that gives an explicit isometric interpretation to the tensor network. 
\end{proof}

We focused here on the case where we have a nontrivial code subspace, for which the isometry condition can be stated. One can also
consider similarly the case where we just have a single holographic state.

\subsection{Entanglement wedges}\label{WedgeSec}

Just like in the HaPPY code, only certain bulk regions in the Bourdon tensor network will satisfy complementary recovery, and hence the Ryu--Takayanagi formula. This has to do with the fact that the isometry group of the tensor network is not quite the whole conformal group.  Here, we give a description of two nice families of such regions, thanks to the notion introduced in \cite{BourdonPaper} of \textit{tree-wall} in the bulk as well as the link structure of the Bourdon building.

Let us summarize the tree-wall construction of \cite{BourdonPaper}. 

\begin{defn}
A \textit{wall} in $I_{p,q}$ is a bi-infinite geodesic contained in the 1-skeleton of $I_{p,q}$. One can then define an equivalence relation on the 1-skeleton of $I_{p,q}$: two edges are equivalent if they share a wall. The equivalence classes are $q$-valent homogeneous trees: they are the \textit{tree-walls} of the building.
\end{defn}

In practice, in order to construct a tree-wall, one can perform the following construction.

\begin{lem}\label{TreeWallLem}
A tree-wall is obtained by the following steps:
\begin{itemize}
\item Choose an edge in the 1-skeleton of $I_{p,q}$.  
\item Add to the tree-wall, on both sides of the chosen edge, the $q-1$ neighboring edges (distance one in the building) that are
at distance $2$ from this edge in the graph of the link. 
\item Repeat the process.
\end{itemize}
\end{lem}
 
\begin{proof}
The link of $I_{p,q}$ is the bipartite graph $K(q,q)$, so in this bipartite graph, $q-1$ vertices (which correspond to $q-1$ edges of the building) are diametrically opposed to our edge (i.e.~at distance $2$ of it on the link's graph).

\end{proof}

\begin{rem}
Note that tree-walls cut the building (and hence its boundary) into $q$ connected components. These connected components will be suitable boundary regions to study entanglement entropy, and the bulk tree-walls will be the analogues of Ryu--Takayanagi surfaces for the Bourdon building. Hence we shall call these connected components entanglement wedges.
\end{rem}

There is another way to look at the tree-wall construction, in terms of the Coxeter system of the building. 

\begin{lem}\label{WedgeLem1}
Consider a given chamber $C$ of the Bourdon building, and pick one of its edges, say $E$. The entanglement wedge with tree-wall boundary
associated to the choice of $C$ and $E$ is obtained by considering, in all apartments containing $C$, 
the portion that is on the same side as $C$ with respect to the hyperplane determined by $E$. 
\end{lem}

\begin{proof}
In any apartment containing $C$, $E$ defines a reflection with respect to a given hyperplane of this apartment. The intersection of the entanglement wedge defined by $C$ and $E$ with the apartment then corresponds to the portion of the apartment on the same side of the hyperplane as $C$. By applying this construction to all apartments containing $C$, we recover the same entanglement wedge, with a tree-wall boundary.
\end{proof}

The following procedure describes another geometric method for the construction of valid entanglement wedges in $I_{p,q}$.

\begin{lem}\label{WedgeLem2}
The following construction defines a valid entanglement wedge in $I_{p,q}$.
\begin{itemize}
\item Choose a vertex $V$ in the building. 
\item Look at the link around the chosen vertex and pick an edge in this link, which corresponds to a chamber $C$ in $I_{p,q}$. 
\item In a given apartment $A$, consider the set $S(C,V,A)$ of tiles defined by the chamber $C$ and the vertex $V$ as the set of tiles containing $C$ and delimited by the two hyperplanes which are edges of the chamber and intersect at $V$.
\item Define the entanglement wedge associated to $C$ and $V$ as the union of all $S(C,V,A)$, for $A$ containing $C$. 
\end{itemize}
In this case, the entanglement wedge is delimited by two half-tree-walls attached to each other.
\end{lem}

Figure \ref{fig:wedges} shows the intersection of the two types of entanglement wedges described in this section with an apartment of the Bourdon building.

\begin{figure}
    \centering
    \includegraphics[scale=0.7]{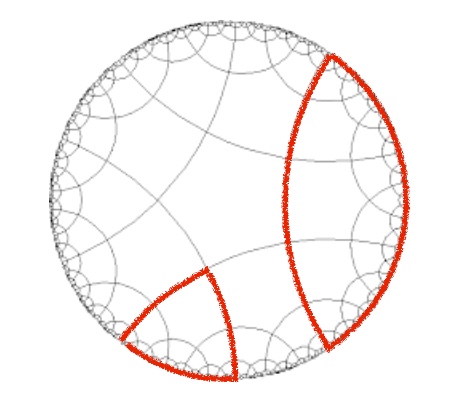}
    \caption{The intersections of the two types of entanglement wedges described here with an apartment of the Bourdon building. Each geodesic should be seen as a portion of a tree-wall in the full building.}
    \label{fig:wedges}
\end{figure}

\subsection{Complementary recovery and a Ryu--Takayanagi formula}

We now show that for an entanglement wedge defined in one of the previous manners, complementary recovery is satisfied in the Bourdon tensor network. 

\begin{prop}\label{RecoveryProp}
Complementary recovery holds for an entanglement wedge in $I_{p,q}$ constructed as in Lemmas \ref{TreeWallLem} and \ref{WedgeLem2}.
\end{prop}

\begin{proof}
By symmetry, we will assume without loss of generality, both in the case defined by the intersection of two tree-walls and the case defined by a single tree-wall, that the considered entanglement wedges are on the opposite side of the tree-walls from the center of the building. These regions can be identified with regions in the network that are spanned by semi-infinite geodesics that start from the center and pass through a given chamber $C$ (see Remark \ref{ball}). As shown in \cite{Pastawski:2015qua}, complementary recovery amounts to showing that the greedy algorithm reaches the surface which delimits the entanglement wedge, both starting from the boundary of the entanglement wedge (i.e.~the subset of the Gromov boundary of $I_{p,q}$ that can be reached by geodesics inside the entanglement wedge), and its complement. Then, our proof amounts to showing that at layer $n$, the chambers inside the entanglement wedge share an edge with at least as many chambers at layer $n+1$ as at layer $n$ or $n-1$. This comes from the apartment structure of the Bourdon building: each chamber can be included in an apartment which is isomorphic to a regular tiling of the hyperbolic disk, and for which each tile at layer $n$ is in contact with at least as many at layer $n+1$ as at layer $n$ or $n-1$. By symmetry of the branching, the other tiles sharing the same edge will also be at layer $n+1$, which finishes the proof of the fact that it is possible to reconstruct the entanglement wedge on its boundary.

Now, it is also possible to reconstruct the complement of the entanglement wedge on the complementary boundary region. In the case in which the entanglement wedge corresponds to a region delimited by a tree-wall, the $q$ connected components are symmetric with respect to the tree-wall, and thus it is possible to achieve reconstruction on the $q-1$ connected components by symmetry.

Similarly, edges in the link (corresponding to chambers in the building) around a given vertex have entanglement wedges that form a partition of the bulk, and satisfy recovery. Thus, 
the complement of the entanglement wedge defined by a given chamber and a link vertex also satisfies recovery.
\end{proof}

\begin{prop}
The Ryu--Takayanagi formula holds for entanglement wedges in $I_{p,q}$ constructed as in Lemmas \ref{TreeWallLem} and \ref{WedgeLem2}.
\end{prop}

\begin{proof}
Given the previous Proposition~\ref{RecoveryProp}, the Ryu--Takayanagi formula directly follows from the argument of \cite{Pastawski:2015qua}: 
since the entanglement wedge satisfies complementary recovery, the entanglement entropy of a boundary state will satisfy 
$$S_{phys}(\rho)=N_{cut}\,\mathrm{log}\,d+S_{code}(\rho),$$where $N_{cut}$ is the number of links that are cut by the boundary of the entanglement wedge and $d$ is the bond dimension. 
\end{proof}

\subsubsection{Ball entanglement wedges}

We now turn our attention to entanglement wedges that correspond to ball-shaped regions on the boundary. Their construction goes as follows.

\begin{lem}\label{BallWedgeLem}
The following procedure gives an entanglement wedge construction in $I_{p,q}$.
\begin{itemize}
\item Fix a chamber $C$ in $I_{p,q}$ and the central chamber $O$. 
\item The ``ball entanglement wedge" defined by $C$ and $O$ is then given by the set of semi-infinite geodesics on the tensor network starting from $C$ that can be extended to a semi-infinite geodesic starting from $O$. 
\end{itemize}
\end{lem}

\begin{rem}
\label{ball}
Note that ball entanglement wedges are particular cases of the entanglement wedges defined in Subsection~\ref{WedgeSec}, and can be defined by a chamber and either an edge or a vertex depending on where the tile is in the building. In particular, the wedges we use in the proof of Theorem \ref{RecoveryProp} can be described in terms of ball entanglement wedges.
\end{rem}

There is a nice way to estimate $N_{cut}$ in terms of the radius of the ball which is the boundary of our entanglement wedge. 

\begin{prop}
Let $\Lambda$ be the layer of the network at which the code is cut off. Consider a bulk ball entanglement wedge $R$ starting at a point $z$, and let $g$ be the distance of the point $z$ to $O$ (note that $g$ is the Gromov product of any two ends of the boundary of the entanglement wedge on the tensor network.) Then, if $q\geq 3$, the number of tile edges in $R$ at layer less than $\Lambda$ which are contained in $\partial R$ equals $$N_{tiles}=\frac{2}{q-2}((q-1)^{\Lambda-g+1}-1)-1$$ if the entanglement wedge is delimited by a tree-wall, and $$N_{tiles}=\frac{2}{q-2}((q-1)^{\Lambda-g+1}-1)$$ if the entanglement wedge is delimited by two half-tree-walls.
\end{prop}

\begin{proof}
This is a direct consequence of the tree structure of the boundary of the entanglement wedge.
\end{proof}

At each of these tiles, $q-1$ links are cut. Therefore, $$N_{cut}=(q-1)\left(\frac{2}{q-2}((q-1)^{\Lambda-g+1}-1)-1\right),$$ or $$N_{cut}=(q-1)\frac{2}{q-2}((q-1)^{\Lambda-g+1}-1).$$
Hence, $$N_{cut}\sim \frac{2}{q-2}\frac{(q-1)^{\Lambda+2}}{(q-1)^g}.$$

We now follow arguments of \cite{BourdonPaper} to obtain a precise Ryu--Takayanagi formula for ball-shaped regions of the boundary. 

Let us consider the case $q>2$. 
Let $\text{Hdim} \,\partial I_{p,q}$ denote the Hausdorff dimension $\text{Hdim} \,\delta_x$ of Theorem~1.1 and 1.2a of \cite{BourdonPaper} (see Lemma 3.1.4 of \cite{BourdonPaper}).

\begin{thm}
The Ryu--Takayanagi formula holds for ball entanglement wedges in $I_{p,q}$, as in Lemma~\ref{BallWedgeLem}, with
\begin{equation}\label{NcutBall}
C^{-1}r^\beta\leq \frac{N_{cut}}{(q-1)^\Lambda}\leq Cr^\beta,
\end{equation} 
for a constant $C>0$ (independent of the boundary region), and with 
\begin{equation}\label{betadim}
\beta={\rm Hdim}\, \partial I_{p,q}-1,
\end{equation}
with $r$ the radius of the boundary ball. 
\end{thm}

\begin{proof}

 In \cite{BourdonPaper}, it is proven that there exists a visual metric $\delta_x$ of parameter $e^{\tau(p,2)}$ on the boundary of the Bourdon building, where $\tau(p,2)$ is the growth rate of the Weyl group of the building. By definition of the distance $\delta_x$, the radius of a boundary ball is controlled by $a^{-g}=e^{-g\tau(p,2)}$.\footnote{One can make this statement more rigorous by using Sullivan's shadow lemma, described in Appendix \ref{PStheory}, to relate the geometry of the group $\Gamma$ to that of the Bourdon building itself. See also Section 3.2 of \cite{BourdonPaper}.} Note that then, \begin{equation}\frac{N_{cut}}{(q-1)^\Lambda}\sim \frac{1}{(q-1)^g}\sim r^\beta,\end{equation} where \begin{equation}\beta=\frac{\mathrm{log}(q-1)}{\tau(p,2)}.\end{equation} For this specific choice of $a=e^{\tau(p,2)}$, it can be shown \cite{Coornaert} that \begin{equation}\beta= \mathrm{Hdim}\,\delta_x-1.\end{equation} This is because in the case of the Bourdon building, we have \cite{BourdonPaper}\begin{equation}\mathrm{Hdim}\,\delta_x= \frac{\tau(p,q)}{\tau(p,2)},\end{equation} and \begin{equation}\tau(p,q)=\tau(p,2)+\mathrm{log}(q-1).\end{equation}

 \end{proof}
\begin{rem}
This proof implies, among other things, that this choice of $\delta_x$ realizes the \textit{conformal dimension} of the boundary of the Bourdon building \cite{Coornaert}. This is an important result from the point of view of geometric group theory.
\end{rem}
 This result is quite striking: it tells us that entanglement entropy in our tensor network knows about the Hausdorff dimension of the boundary! It is also nice to realize that this behavior is in agreement with cases of holographic CFTs whose Cauchy slice dimension is an integer strictly larger than 1: in this case it is known that CFT entanglement entropy scales as a power law in the radius, where the exponent is dictated by the dimension of the time slice of the boundary minus one. Therefore, it seems like our tensor network is simulating a conformal field theory on a fractal! Ryu--Takayanagi surfaces in the network are given by rooted trees, and the entanglement entropy for corresponding ball-shaped regions are given by our power law.

\subsection{The case of HaPPY-like tilings}

The case $q=2$ corresponds to a Fuchsian tiling of the hyperbolic plane, i.e.~to a $p$-gon HaPPY tiling. In this case, it is easier to link $N_{cut}$ with the size of the boundary ball. Indeed, we have $$N_{tiles}=N_{cut}=2(\Lambda-g)+1,$$ and the size $r$ of the boundary region satisfies $$-g\sim\mathrm{log\;}{r}.$$ Hence $$N_{cut}\sim\mathrm{log}\frac{r}{\varepsilon},$$where $$\varepsilon=e^{-\frac{\Lambda}{\alpha}}$$ for some $\alpha>0$. This approximately reproduces the logarithmic behavior of entanglement entropy on the boundary.

\subsection{The infinite-dimensional limit: Hilbert spaces and nets of local algebras}\label{InfDimSec}

One of the advantages of our construction is that it provides us with a large family of nets of infinite-dimensional exact quantum error-correcting codes with complementary recovery. It is therefore an explicit example of nets of conditional expectations, as introduced by Faulkner in \cite{Faulkner:2020hzi}.

Let us first associate an infinite-dimensional code and a physical Hilbert space to our tensor network. 
The idea will be to take a direct limit of Hilbert spaces, both in the bulk and on the boundary. 

\begin{prop}\label{HilbertInfDimProp}
There is an injection $\mathcal{H}^\Lambda_{code}\rightarrow \mathcal{H}^{\Lambda+1}_{code}$ of Hilbert spaces from the truncated network at layer $\Lambda$
to level $\Lambda+1$ compatible with the maps $u^\Lambda:  \mathcal{H}^\Lambda_{code}\to \mathcal{H}^\Lambda_{phys}$ through commutative diagrams. These 
maps define as direct limits the infinite dimensional Hilbert spaces 
\begin{equation}\label{Hinfdim}
\mathcal{H}_{code} =\varinjlim_\Lambda \mathcal{H}_{code}^\Lambda \ \ \ \ \text{ and } \ \ \ \  \mathcal{H}_{phys} =\varinjlim_\Lambda \mathcal{H}_{phys}^\Lambda
\end{equation}
with an induced isometry $u:\;\mathcal{H}_{code}\rightarrow\mathcal{H}_{phys}$.
\end{prop}

\begin{proof}
For the bulk Hilbert space, define a reference state for the dangling qudits. Let us denote it by $\ket{\text{ref}}$. If $\mathcal{H}^\Lambda_{code}$ denotes the code subspace of the truncated network at layer $\Lambda$, we define an injection $\mathcal{H}^\Lambda_{code}\rightarrow \mathcal{H}^{\Lambda+1}_{code}$ by tensoring the state of $\mathcal{H}^\Lambda_{code}$ with qudits in the state $\ket{\text{ref}}$ at layer $\Lambda +1$.\footnote{These maps have a few shortcomings, like the fact that they do not create an entangled bulk state. See \cite{Gesteau:2020hoz} for another possible choice of Hilbert space maps. However, our maps here have good functorial properties at the level of operators, and will be enough for our purposes.} In order to construct a map from $\mathcal{H}^\Lambda_{phys}$ to $\mathcal{H}^{\Lambda+1}_{phys}$, we simply take the state at layer $\Lambda$, and map it through the tensor network from layer $\Lambda$ to layer $\Lambda+1$, with all dangling bulk nodes fixed in the state $\ket{\text{ref}}$. If $u^\Lambda$ denotes the map from $\mathcal{H}^\Lambda_{code}$ to $\mathcal{H}^\Lambda_{phys}$, we then obtain a commutative diagram of the form
\begin{align}
\begin{array}{ccccccc}
	\calh_{code}^1 & \rightarrow & \calh_{code}^2 & \rightarrow & \quad\cdots \\ 
	\downarrow &  & \downarrow &  & \\ 
	\calh_{phys}^1 & \rightarrow & \calh_{phys}^2 & \rightarrow & \quad\cdots & \quad,
\end{array} 
\label{eq:commutativediagram}
\end{align}
where the horizontal arrows denote the bulk-to-bulk maps and boundary-to-boundary maps, and the vertical arrows denote the isometries $u^\Lambda$. We can then take the direct limit of this diagram and define $$\mathcal{H}_{code}=\underset{\underset{\Lambda}{\longrightarrow}}{\mathrm{lim}}\;\mathcal{H}_{code}^\Lambda,$$and $$\mathcal{H}_{phys}=\underset{\underset{\Lambda}{\longrightarrow}}{\mathrm{lim}}\;\mathcal{H}_{phys}^\Lambda,$$ as well as an isometry $$u:\;\mathcal{H}_{code}\rightarrow\mathcal{H}_{phys}$$
as in \eqref{Hinfdim}.
\end{proof}

Now, we are interested in studying a net of local observables on the boundary, and assigning an entanglement wedge to each of them. 

\begin{thm}
For a given entanglement wedge associated to a tile $T$ as before, there are $C^*$-algebras $\mathcal{A}_{code}$ and $\mathcal{A}_{phys}$, obtained as direct limits of entanglement wedge algebras $\mathcal{A}_{code}^\Lambda$ 
at layers $\Lambda$ and their bulk-to-boundary maps. They are related by a unital isometric $\star$-homomorphism $\iota: \mathcal{A}_{code}\to \mathcal{A}_{phys}$,
which is compatible with the isometry $u:\;\mathcal{H}_{code}\rightarrow\mathcal{H}_{phys}$ of Proposition~\ref{HilbertInfDimProp} in the sense that
$\iota(A)u = u A$.
\end{thm}

\begin{proof}
For this, consider a bulk tile $T$, and define an entanglement wedge associated to this tile as in Subsection~\ref{WedgeSec}. On top of each bulk qudit, we introduce a finite-dimensional algebra $\mathcal{M}_d(\mathbb{C})$. The entanglement wedge algebra $\mathcal{A}_{code}^\Lambda$ at layer $\Lambda$ is given by the tensor product of all bulk dangling legs, and the map from $\mathcal{A}^\Lambda$ to $\mathcal{A}^{\Lambda+1}$ is given by tensoring with the identity on qudits at layer $\Lambda+1$. Recall that there is a well-defined bulk-to-boundary map $\iota^\Lambda$ at the level of operators, by successively tensoring with the appropriate number of identities and conjugating by perfect tensor unitaries \cite{Pastawski:2015qua},\footnote{This map is not unique.} and that by the previous subsections, this map has a range contained on the boundary of the entanglement wedge defined by $T$. Moreover, one can define a map from the complement algebra to the complement boundary region, by complementary recovery. We can use the $\iota^\Lambda$ to define a boundary-to-boundary map at the level of algebras: just add one more layer of tensor network, and conjugate an operator by the perfect tensor isometries with identitity matrices on the new bulk nodes, following the map given by $\iota^{\Lambda+1}$. We then obtain another commutative diagram of the form
\begin{align}
\begin{array}{ccccccc}
	\mathcal{A}_{code}^1(T) & \rightarrow & \mathcal{A}_{code}^2(T) & \rightarrow & \quad\cdots \\ 
	\downarrow &  & \downarrow &  & \\ 
	\mathcal{A}_{phys}^1(T) & \rightarrow & \mathcal{A}_{phys}^2(T) & \rightarrow & \quad\cdots & \quad,
\end{array} 
\label{eq:commutativediagram2}
\end{align}
where the horizontal arrows are the bulk-to-bulk and boundary-to-boundary maps, and the vertical arrows are the $\iota^\Lambda$. Like in the previous case, we can take the direct limit $C^\ast$-algebra, and we obtain two $C^\ast$-algebras $\mathcal{A}_{code}$ and $\mathcal{A}_{phys}$, related by a unital isometric $\ast$-homomorphism $$\iota:\;\mathcal{A}_{code}\rightarrow\mathcal{A}_{phys}.$$ Moreover we can see $\mathcal{A}_{code}$ and $\mathcal{A}_{phys}$ as acting on $\mathcal{H}_{code}$ and $\mathcal{H}_{phys}$, and by construction of $\iota$, for $A\in\mathcal{A}_{code}$, $$\iota(A)u=uA,$$ and similarly for the complementary algebras.
\end{proof}

Note that this construction depends on our choice of map $\iota$, which is not always unique. For example, it is not unique in the case of the pentagonal HaPPY code. This breaks the symmetry of the network, but it still gives rise to a net of holographic conditional expectations in the sense of \cite{Faulkner:2020hzi} (with the difference that we left the construction here at the level of $C^\ast$-algebras).

\section{The general case: holographic tensor networks on hyperbolic buildings}\label{HighDimSec}

Our construction for the case of Bourdon buildings can be generalized to a much larger class of buildings in various integer bulk dimensions and non-integer boundary Hausdorff dimension.
We first give a set of sufficient conditions for our construction to easily generalize. We then introduce a few interesting examples of tensor networks that satisfy these conditions.

\subsection{A class of tensor networks}

We want to extend our construction to a well-chosen class of hyperbolic buildings. In order for the same method to work, we need to check the following conditions:

\begin{itemize}
    \item The tensor network still defines an isometric map at each layer.
    \item Complementary recovery still works for well-chosen entanglement wedges. 
    \item The Ryu--Takayanagi scaling of the entanglement entropy for well-chosen boundary balls still follows a power law of exponent $\beta-1$, where $\beta$ encodes the dimension of the boundary time slice, if $\beta>1$, or a logarithmic behavior if $\beta=1$.
\end{itemize}

First, in order to show that the tensor network defines an isometric map at each layer, we had to use the fact that, given a central chamber and a fixed apartment containing it, any chamber of that apartment is adjacent to more outgoing tiles further away from the center than tiles closer or equidistant to the center. Thus, our argument only used the apartment structure. We are therefore reduced to finding a condition of the Weyl group of our building that guarantees that it maps the bulk to the boundary Hilbert spaces isometrically. 
An explicit sufficient condition on the Coxeter system has been found by Kohler and Cubitt (see Section~6.1.2 of \cite{Kohler:2018kqk}), and we use it here:

\begin{defn} Let $B$ be a building of Weyl group $W$ with Coxeter system $(W,S)$. Let $\mathcal{F}:=\{J\subset S, W_J\;\text{is finite}\}.$ We will say that $B$ satisfies the \textit{isometry condition} if for all $J\in\mathcal{F}$, $$|J|\leq\frac{t-2}{2},$$ where $t$ is the number of indices of the perfect tensor, divided by the branching.
\end{defn}

\begin{lem}\label{IsomRightAngleLem}
The isometry condition is always satisfied for a right-angled hyperbolic building. 
\end{lem}

\begin{proof}
Indeed, in that case the Coxeter system associated to the building has matrix entries $2$ and $\infty$. Hence, the subsets $J$ of $S$ for which $W_J$ is finite can only contain elements for which the Coxeter matrix elements for each pair are $2$. These must then be adjacent faces in the basic chamber of the building. Now, a hyperbolic polytope can be right angled iff it has more faces than a hypercube. This means that a maximal set of adjacent faces will always have less elements than half of the number of faces of the polytope. Identifying Coxeter generators with the faces of the polytope, and recalling that there is always at least one dangling leg inside the perfect tensors, we obtain the isometry condition.
\end{proof}

\begin{rem}
When the isometry condition is satisfied on top of our other conditions, a slight adaptation of the argument of \cite{Pastawski:2015qua} shows that our building defines an isometry from the bulk to the boundary at all layers.
\end{rem}

The second point that needs to be confirmed is complementary recovery for well-chosen entanglement wedges. 

\begin{lem}
The entanglement wedge constructions of Section~\ref{WedgeSec} extend to arbitrary hyperbolic buildings satisfying the isometry condition and satisfy complementary recovery.
\end{lem}

\begin{proof}
This is guaranteed by the building structure of our networks. More precisely, consider a link in the building, and an $(n-1)$-polytope $P$ on that link corresponding to a given tensor in the network. Embed $P$ into a given apartment of the network. This apartment is isomorphic to a hyperbolic Coxeter system, hence we can repeat the construction of the previous section and use the link to construct a partition of the bulk into reconstructable regions. The other definition in terms of one single hyperplane (generalization of the tree-wall) is even more straightforward. 
\end{proof}

The third and last point is probably the most subtle one: our argument on the Hausdorff dimension scaling required calculating the Hausdorff dimension of the Bourdon building for a specific visual metric. A full study is out of the scope of this paper, but in order to be able to generalize it, we want our building chambers to each connect to the same number of edges through a wall, and to use the transitivity of the action of the isometry group of $X$ on the set of apartments that contain it, to reproduce the proof of the Ryu--Takayanagi scaling. This requires $B$ to contain a chamber whose fixator in the isometry group of $X$ acts transitively on the set of apartments that contain it. We now obtain the following result:

\begin{thm}\label{nHypThm}
Let $B$ be an $n$-dimensional hyperbolic building such that:
\begin{itemize}
    \item The Weyl group of $B$ satisfies the isometry condition.
    \item The link at each point is the same, and each $(n-1)$-polytope of the link is $q$-valent (in the sense that it touches $q$ other $(n-1)$-polytopes) for some $q$.
    \item $B$ contains a chamber whose fixator in the isometry group of $X$ acts transitively on the set of apartments that contain it.
\end{itemize}
Then, $B$ defines a quantum error-correcting code with complementary recovery for well-chosen bulk regions, and in a large class of these tensor networks, for these regions, the size of the Ryu--Takayanagi surface scales like $r^{\beta-1}$, where $r$ is the radius of the associated boundary ball, and $\beta$ is the scaling dimension of $\partial B$.
\end{thm}

\begin{rem}
Note that in this theorem, the Ryu--Takayanagi formula scales like the \textit{scaling dimension} (self-similarity dimension, see Chapter~4 of \cite{PJS}) of the boundary, but not necessarily like the Hausdorff dimension of a given visual metric. This is because in order for these two dimensions to be equal, one needs to show the existence of a visual metric with such a Hausdorff dimension, and this existence property is not always guaranteed. In the case of the Bourdon building, it was shown \cite{BourdonPaper} that this could be done by explicitly constructing a visual metric with parameter $a$ equal to the growth rate of the Weyl group of the building, that realizes the conformal dimension. However, it is a difficult and important problem in geometric group theory to understand when such a metric can be constructed in higher dimension. In particular, we expect this question to be related to subtle rigidity properties of the buildings. Even in order to show a matching of the scaling dimensions, one needs to have some nice formula that relates the growth rates of the building and of its apartments. See for example \cite{Clais}, where some partial answers to these questions are given in the case of right-angled buildings, particularly in three dimensions. We leave a more precise study of the scaling properties of the RT surfaces in higher-rank buildings, as well as of these issues related to the comparison of Hausdorff, scaling and conformal dimensions, to future work.
\end{rem}

\subsection{Higher dimensional examples}

We now show that our techniques can be adapted to construct holographic codes in arbitrary integer dimensions (as well as non-integer boundary Hausdorff dimensions). This will be done through an explicit construction, due to Davis--Moussong, of an interesting class of hyperbolic buildings in any given dimension. This construction is slightly more technical than the rest of the paper, the main takeaway being that holographic codes exist in all integer bulk dimensions, and that the Ryu--Takayanagi formula and entanglement wedge reconstruction for well-chosen regions, carry over to these more general cases.

We have focused primarily on the Bourdon buildings, which are the primary example of hyperbolic buildings. More generally, it is known
that hyperbolic buildings are more difficult to obtain than Euclidean ones. Indeed, if one takes as part of the definition of buildings the
requirement that they are polyhedral complexes where apartments are (in the hyperbolic case) polyhedrally isometric to a tesselation of
the $n$-dimensional hyperbolic space, then there are strong restrictions on the dimension of hyperbolic buildings.  There is a bound
$n\leq 29$ on the dimension of a compact convex hyperbolic Coxeter polytope \cite{Vinberg}.
In particular, Theorem~\ref{nHypThm},
as stated, only applies in this range. 

In the case of the Bourdon buildings, as we have seen, a useful property is the fact that they are right-angled buildings. 
In general, a Coxeter system $(W,S)$ is right-angled  if all the $m_{ij}$ with $i\neq j$ in the relations are equal to $2$ or $\infty$.
With this further requirement, it is known that right-angled Coxeter polytopes can only exist in dimension $n\leq 4$ \cite{PotVin}.
Thus, right-angled hyperbolic buildings (with the definition as above) can only exist in dimension $n\leq 4$. 

This seems to limit the range of validity of the general setting we introduced above for the construction of holographic
tensor networks on hyperbolic buildings. However, it is in fact possible to relax slightly the definition of buildings in such
a way that right-angled hyperbolic buildings will be available in arbitrary dimension. This was done in \cite{JanSw}
by considering geometries where the apartments are Davis--Moussong complexes of the Coxeter group $W$, instead of copies
of hyperbolic space tessellated by the action of $W$. In general, the Davis--Moussong complex associated to a Coxeter system
$(W,S)$ is a piecewise Euclidean (non-positively curved) 
polyhedral complex with a properly discontinuous cocompact action of $W$, see
\cite{Davis, Mouss}.

The construction of the Davis--Moussong complex $K(W,S)$ is obtained in the following way~\cite{Mouss}.
Given a Coxeter system, the associated nerve $N(W,S)$ is the simplicial complex with one
vertex for each element of $S$, with subset $T\subset S$ defining a simplex of $N(W,S)$
if the subgroup $W_T\subset W$ generated by $T$ is finite. 
Let $N'$ be the barycentric subdivision of $N=N(W,S)$ and let $CN'$ be the cone of $N'$. 
For each $s\in S$ let $X_s$ be the closed star of the vertex $s$ in $N'$. The collection
$M=\{ X_s \}_{s\in S}$ of closed subspaces of $X=CN'$ is called the set of panels. For $x\in X$ 
one also sets $S(x):=\{ s\in S\,|\, x\in X_s \}$.
There is an associated a universal $W$-space $U(W,X,M)$, with a CW complex structure, 
given by the quotient of $W\times X$
by the equivalence relation $(w,x)\sim (w'x')$ if $x=x'$ and $w^{-1} w'\in W_{S(x)}$. 
This space is universal with respect to maps $f: X\to Y$ with $s f(x)=f(x)$ for $x\in X_s$:
each such map uniquely extends to a continuous $W$-equivariant map $f: U\to Y$.
The Davis--Moussong complex is $K(W,S)=U(W,CN',M)$. It is contractible with $K(W,S)/W\simeq CN'$. 

There is a characterization of hyperbolic Coxeter groups, in terms of the presence of a 
hyperbolic structure on the Davis--Moussong complex $K(W,S)$. The Coxeter group is
hyperbolic iff it contains no subgroup isomorphic to ${\mathbb Z} \oplus {\mathbb Z}$.
This condition is in turn equivalent to the condition that there is no subset $T\subset S$
with $W_T$ an affine Coxeter system of rank $\geq 3$ and there are no pairs $T_1,T_2$ of
disjoint subsets of $S$ for which $W_{T_1}$ and $W_{T_2}$ commute and are infinite, 
see Theorem 17.1 of \cite{Mouss}. This last condition in turn implies that $K(W,S)$ can be given
a hyperbolic structure by considering for each $T\subset S$ with $W_T$ finite a hyperbolic space
${\mathbb H}^{n_T}$ with $n_T=\# T$. The building blocks of $CN'$ are the $CN(W_T,T)'$
for $T\subset S$ with $W_T$ finite. Each of these building blocks is homeomorphic to a
combinatorial $n_T$-cube $B(W_T)_\epsilon\subset {\mathbb H}^{n_T}$, for some $\epsilon>0$
(see \cite{Mouss} for more details). These building blocks are glued
together according to the relations of subsets $T\subset S$ into a decomposition of $K(W,S)$. 
The condition above on the hyperbolicity of the Coxeter group $W$ ensures that, with this
hyperbolic structure on the blocks $CN(W_T,T)'$, the complex $K(W,S)$ itself has a the
structure of a hyperbolic complex. 

The difference with the usual apartments of buildings in the more restrictive sense is that 
here the contractible manifolds $K(W,S)$ are in general not homeomorphic to Euclidean 
space (hence not tessellated copies of ${\mathbb H}^n$). Indeed, the boundary at infinity of $K(W,S)$
is not necessarily simply connected, but it has the topology of a generalized $(n-1)$-dimensional 
homology sphere, see \cite{Davis} and \cite{JanSw}.

The construction of \cite{JanSw} of higher dimensional right-angled hyperbolic buildings with
Davis--Moussong complexes as apartments is obtained via complexes of groups. 

A complex of groups is an assignment of groups and compatible maps to a simplicial
complex that reflects the properties of the orbit space of a group action on a cell complex, see \cite{Corson}.
It generalizes the Bass--Serre construction of graphs of groups. 

A complex of groups $G(K)$ consists of combinatorial CW complex $K$ 
(namely a CW complex that is either simplicial or that can be subdivided into simplicial complexes), 
with a group $G_{e_\alpha}$ assigned to each cell $e_\alpha$ and monomorphisms $\phi_{\alpha,\beta}: 
G_{e_\alpha} \to G_{e_\beta}$ for each cell $e_\beta$ in the boundary of $e_\alpha$. Boundary inclusions
$e_\gamma \subset e_\beta \subset e_\alpha$ give $Ad(g)\phi_{\alpha,\gamma}=\phi_{\beta, \gamma}\circ \phi_{\alpha,\beta}$, for some $g\in G_{e_\gamma}$ acting by conjugation. For our purposes we can assume that
$K$ is a polyhedral complex (or a simplicial complex after passing to barycentric subdivision). 

The complex of groups associated to a simplicial action of a group on a combinatorial cell complex has
finite stabilizer groups attached to the cells, with monomorphisms of stabilizers contravariantly 
associated to inclusions of cells. A complex of groups is developable if it is the
complex of groups associated to a simplicial action of a group on a simply connected 
combinatorial cell complex. Not all complexes of groups are developable, but developability
is implied by a non-positively curved condition \cite{BriHae}. 

The construction via complexes of groups of a right-angled building with apartments shaped as 
Davis--Moussong complexes $K(W,S)$ is obtained as follows. Start with a right-angled Coxeter
system $(W,S)$ that satisfies the hyperbolicity condition above, so that $K(W,S)$ has a hyperbolic
structure. Take as additional datum a set $\{ q_s \}_{s\in S}$ of integers $q_s\geq 2$, and let $G_s$
be a group of order $q_s$. 
As above, vertices of $CN'$ has type some $J\subset S$ with $W_J$ finite.  Let
$G(K)$ be the complex of groups that assigns to a vertex of type $J$ the group given by
the direct product $G_J=\amalg_{s\in J} G_s$, with maps given by inclusions. This complex
of groups is developable with cover a right-angled building. 
As above let $X_s$ be the closed star of the vertex $s$ in $N'$. Each copy of $X_s$ is
contained in $q_s$ chambers in this building, with each chamber given by a copy of $CN'$.  

The existence in any dimension of this type of hyperbolic buildings with Davis--Moussong apartments 
is then proved in \cite{JanSw} by showing the existence in any dimension of a right-angled Coxeter system
$(W,S)$ satisfying the hyperbolicity condition, i.e. containing no subgroup isomorphic 
to ${\mathbb Z} \oplus {\mathbb Z}$. 

This class of buildings satisfy our conditions for tensor networks with good holographic properties. 

\begin{thm}\label{DMbuildings}
The hyperbolic buildings built using Davis--Moussong complexes and with all the $q_s=q$, satisfy the isometry condition 
and have complementary recovery.\end{thm}

\begin{proof}
The isometry condition is still satisfied: as shown in \cite{Mouss}, the girth of the links is strictly greater than $2\pi$, which means that if one fixes a vertex of a chamber of the building, as the chamber is right-angled, more faces of the chamber will not touch that vertex than touch it, and $W_J$ can be finite only if $J$ only contains faces that touch the same vertex. This same argument allows us to prove complementary recovery for well-chosen regions (generalizations of either tree-walls or vertex-based entanglement wedges). The third condition formulated in Theorem~\ref{nHypThm} is satisfied by construction: the Davis--Moussong complex can be seen as the quotient of the building by the action of the group developed by the complex of stabilizers. So by acting on an apartment by all the group elements that stabilize a given cell (which makes sense because the action is simplicial), one can obtain all apartments containing that cell. 
\end{proof}

\begin{rem}In the proof, we also included a discussion of the third condition of Theorem~\ref{nHypThm}, and it should also be possible to formulate a nice statement about of the area of the Ryu--Takayanagi surfaces in terms of the scaling, Hausdorff or conformal dimensions of the boundary. This question should also be interesting from the point of view of geometric group theory.
\end{rem}

In particular, this shows the existence of tensor networks with good holographic properties in arbitrary dimension.

\section{Discussion}
\label{DiscussSec}
In this paper, we explained how one could construct a large class of holographic tensor networks from hyperbolic buildings. The language of buildings and Gromov-hyperbolicity allows to recover the usual properties of hyperbolic tensor networks and to describe them in a unified way. In particular, our buildings:
\begin{itemize}
\item Contain a large class of bulk regions that satisfy complementary recovery. These regions admit an explicit description in terms of building theory.
\item Satisfy the Ryu--Takayanagi formula. For ball-shaped boundary regions of Hausdorff dimension strictly greater than 1, the entanglement entropy of holographic states follows a power law in the radius of the ball, with exponent given by the \textit{Hausdorff dimension of the boundary minus one}, in a large number of cases including the one of Bourdon buildings. If the boundary has dimension 1, we recover the logarithmic scaling of the HaPPY code.
\item Exist for boundaries of all integer dimensions, and therefore provide explicit examples of holographic codes in all integer dimensions.
\item Recover all known holographic codes constructed out of networks of perfect tensors as particular cases. In particular, the HaPPY code as well as the higher-dimensional examples of \cite{Kohler:2018kqk} can be studied through the lens of our construction.
\end{itemize}

Several future directions can be envisioned. First, we only considered holographic codes made out of perfect tensors in this paper, but it would be nice to study holographic codes made out of random tensors in our context. 

It would also be interesting to understand the situation for disconnected boundary regions better. Our explicit entanglement wedge constructions all involve a connected boundary region, and it would be nice to understand whether similar descriptions based on building theory hold for disconnected boundary regions.

Another question is whether the fact that the entanglement entropy of ball-shaped regions sees the fractal dimension of the boundary means that it could be possible to define theories that resemble conformal field theory on a fractal background (see \cite{Boyle:2018uiv} for a related attempt in one dimension).

It is also interesting to ask if more general objects than hyperbolic buildings are well-adapted to the construction of holographic codes. In particular, it would be interesting to understand if quotients of our buildings can be taken in order to describe nontrivial bulk topologies in the spirit of the BTZ-like topologies constructed in \cite{Heydeman:2018qty}. We expect that the setup of latin square designs \cite{CMM} might be helpful to think about this kind of problem.

Finally, this paper showed that the theory of Gromov hyperbolicity can be very useful to show geometric results about the bulk, such as the Ryu--Takayanagi formula. One can then wonder whether the notions of Gromov product and of hyperbolicity can also be utilized in the continuum, to understand the geometric structure of the bulk in full AdS/CFT. 

\section*{Acknowledgements}
This work was partially supported by NSF grants DMS-1707882 and DMS-2104330.

\appendix 

\section{An example: the building $I_{5,3}$}
\label{APPI53}

In this appendix, we explicitly compute the various quantities and regions described in the bulk of the paper in the case of the simplest Bourdon building: the building $I_{5,3}$.

We first describe explicitly the structure of the building $I_{5,3}$.
\begin{itemize}
\item \textbf{Apartment structure:} The apartments of the building $I_{5,3}$ are HaPPY tessellations of the hyperbolic plane: regular, right-angled tessellations by pentagons. Their Schläfli symbols are $\{5,4\}$. Their Coxeter group is generated by five reflections $r_1,r_2,r_3,r_4,r_5$ such that $(r_ir_{i+1})^2=1$ (where the index is understood modulo 5).
\item \textbf{Link:} The link is the bipartite complete graph $K(3,3)$. This means that at each vertex of the building, six edges concur. These six edges can be split into two groups of three. Each edge shares a face only with the three edges of the opposite group.
\end{itemize}

Then, we introduce the growth rates of the Weyl group and the isometry group of the building \cite{BourdonPaper}:
\begin{itemize}
\item \textbf{Growth rate of the Weyl group:} The growth rate of the Weyl group is equal to \begin{equation}\tau(5,2)=\mathrm{Arccosh}\left(\frac{3}{2}\right)\end{equation}
\item \textbf{Growth rate of the isometry group:} The growth rate of the isometry group is equal to \begin{equation}\tau(5,3)=\mathrm{Arccosh}\left(\frac{3}{2}\right)+\mathrm{log}\, 2.\end{equation}
\end{itemize}

In the case of the Bourdon building, it is possible to show \cite{BourdonPaper} that there exist visual metrics on $\partial I_{5,3}$ whose Hausdorff dimension realizes the conformal dimension of the building, and is equal to the ratio of the growth rate of the isometry group by the growth rate of the Weyl group. We then have:

\begin{equation}
    \mathrm{Hdim}(\partial I_{p,q})=1+\frac{\mathrm{log}\, 2}{\mathrm{Arccosh}\left(\frac{3}{2}\right)}
\end{equation}
The second term on the right hand side corresponds to the exponent of the scaling of entanglement entropy.

Finally, we describe ``ball entanglement wedge" bulk regions that satisfy complementary recovery in the building $I_{5,3}$. They are delimited either by a tree-wall or by the intersection of two tree-walls. In the case of $I_{5,3}$, a tree-wall divides the building into three valid entanglement wedges (see Figure~\ref{fig:I53}).

\begin{figure}[thb]
    \centering
    \includegraphics[scale=0.7]{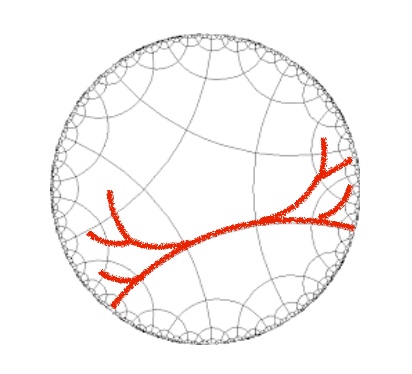}
    \caption{A tree-wall in $I_{5,3}$ and its intersection with an apartment. The tree-wall coincides with a geodesic in the HaPPY apartment, but at every tile, it divides into two branches due to the link structure of $I_{5,3}$. In the end, we end up with a homogeneous trivalent tree.}
    \label{fig:I53}
\end{figure}

\section{Quasi-conformal measures and Patterson--Sullivan theory}
\label{PStheory}
In this slightly more technical appendix, we elaborate on the idea that our networks describe some features of an approximate conformal field theory on a fractal space. The idea is that for any visual metric of the boundary of a Gromov hyperbolic group (such as the isometry group of the Bourdon building, which can be identified with its chambers), one can define a privileged measure that behaves ``almost conformally" under isometries. This measure is called the Patterson--Sullivan measure, and we review its construction here. We begin with a few definitions: 

\begin{defn}
Let $r$ be a geodesic ray in $X$. The Busemann function associated to $r$ is the map \begin{align*}h:\;&X\longrightarrow\mathbb{R}\\&x\longmapsto\underset{t\rightarrow +\infty}{\mathrm{lim}}(|x-r(t)|-t).\end{align*}
\end{defn}
Note that this function is well-defined by the triangle inequality.
\begin{defn}
Let $\Gamma$ be a group of isometries of $X$, let $\gamma\in\Gamma$. For $\xi\in\partial X$, choose a geodesic ray arriving at $\xi$, let $h$ be its Busemann function, and define $$j_\gamma(\xi):=a^{\Delta(\xi)},$$where $$\Delta(\xi)=h(O)-h(\gamma^{-1}O).$$ Let $\mu$ be a regular Borel measure on $\partial X$, of finite nonzero total mass. Then, define $$\gamma^\ast\mu:=\mu\circ\gamma.$$The measure $\mu$ is said to be $\Gamma$-quasiconformal of dimension $D$ if all the $\gamma^\ast\mu$ are absolutely continuous with respect to each other\footnote{This means that they have the same zero-measure sets.} for $\gamma\in\Gamma$, and there exists $C\geq 1$ such that $$C^{-1}j_\gamma^D\leq\frac{d\gamma^\ast\mu}{d\mu} \leq Cj_\gamma^D,$$ $\mu$-almost everywhere.\footnote{That is, the inequality holds everywhere except for a set of measure zero as measured with respect to $\mu$.}
\end{defn}

One can view a $\Gamma$-quasiconformal measure as a measure on $\partial X$ for which $\Gamma$ ``behaves like a conformal group". Under a reasonable assumption on $\Gamma$, there exists a generic construction of such a $\Gamma$-quasiconformal measure, due to Patterson and Sullivan \cite{Coornaert}.

Let $Y$ be the orbit of $O$ under $\Gamma$. For $R\geq 0$, let $n_Y(R)$ be the number of points of $Y$ within a distance $R$ of $O$. 

\begin{defn}
We define the \textit{base}-$a$ \textit{critical exponent of} $\Gamma$ as $$e_a(\Gamma):=\underset{R\rightarrow +\infty}{\mathrm{limsup}}\frac{\mathrm{log}_an_Y(R)}{R}.$$
For $s\geq 0$, we define the Poincaré series $$g_Y(s):=\sum_{y\in Y}a^{-s|y|}.$$ 
\end{defn}

One can then prove \cite{Coornaert} that this series is divergent for $s<e_a(\Gamma)$ and convergent for $s>e_a(\Gamma)$. Then, one can construct a sequence $(s_n)$ of limit $e_a(\Gamma)$, with $s_i>e_a(\Gamma)$. Consider the sequence of measures $$\mu_n:=\frac{1}{g_Y(s_n)}\sum_{y\in Y}a^{-s_n |y|}\delta_y,$$ where $\delta_y$ is the Dirac measure at $y$. As $X\cup\partial X$ is compact, one can then extract a weakly convergent subsequence of $(\mu_n)$. The limit $\mu$ of this subsequence is called a Patterson--Sullivan measure for $\Gamma$ and one can prove (up to some technical refinements on the sequence $(\mu_n)$ in the case where the Poincaré series converges at $e_a(\Gamma)$):
\begin{prop}[\cite{Coornaert}]
If $e_a(\Gamma)$ is finite, then the Patterson--Sullivan measure $\mu$ is $\Gamma$-quasiconformal of exponent $e_a(\Gamma)$. Moreover, its support is the limit set of $\Gamma$, denoted $\Lambda$.
\end{prop}
Under the extra assumption of $\Gamma$ being quasi-convex cocompact \cite{Coornaert}, more can be said about the space of $\Gamma$-quasiconformal measures on $\partial X$:
\begin{prop}[\cite{Coornaert}]
If $\Gamma$ is a quasi-convex cocompact group acting on $X$ such that $e_a(\Gamma)$ is finite, then, $\Lambda$ has Hausdorff dimension $e_a(\Gamma)$, and the associated Hausdorff measure is $\Gamma$-quasiconformal of dimension $e_a(\Gamma)$. Moreover, if $\mu$ is another $\Gamma$-quasiconformal measure whose support is contained in $\Lambda$, then it has dimension $e_a(\Gamma)$ and it is of the form $\psi\mathcal{H}$, where $\mathcal{H}$ is the Hausdorff measure of $\Lambda$ and $\psi\in L_0^\infty(\mathcal{H})$. Reciprocally, all such measures are $\Gamma$-quasiconformal of dimension $e_a(\Gamma)$ with support contained in $\Lambda$.
\end{prop}

A lot of results in Patterson--Sullivan theory rely on Sullivan's shadow lemma. We state it here in its most general form, see also Lemma 8 of \cite{Zhu} for a useful variant:

\begin{prop}[\cite{Coornaert}, Proposition~6.1]
Let $\mu$ be a $\Gamma$-quasiconformal measure of dimension $D$ on $\partial X$. If $O(x,d)$ is the intersection with $\partial X$ of the set of all geodesic rays starting at $O$ and passing at a distance smaller than $d$ of $x$, then there exist constants $C\geq 1$ and $d_0\geq 0$ such that for all $d\geq d_0$ and $\gamma\in\Gamma$, $$C^{-1}r^D\leq \mu(O(x,d))\leq Cr^Da^{2Dd},$$
where we have chosen $x=\gamma^{-1} O$, for the chosen base point $O$ 
and $r=a^{-|x|}$. The set $O(x,d)$ is called the shadow on $\partial X$ of the ball centered on $x$ with radius $d$.
\end{prop}

\end{document}